\documentclass[smallextended]{svjour3}
\usepackage{amssymb}
\usepackage{graphicx}

\usepackage{multirow}
\usepackage{tikz}
\usepackage{mathtools}
\usepackage{mathrsfs}
\usepackage{amsmath}
\usepackage{amssymb}

\usepackage{tikz}
\usetikzlibrary{shapes,arrows,calc}
\usepackage{pifont}
\usepackage[flushleft]{threeparttable}

\newtheorem{defn}{Definition}[subsection]

\begin{document}
\title{Cyclic group based mutual authentication protocol for RFID system}


\author{Pramod Kumar Maurya \and 
               Satya Bagchi 
}


\institute{P. K. Maurya \at
          Department of Mathematics\\ 
          National Institute of Technology Durgapur\\
          Burdwan, India. \\
          \email{pramod$\_$kumar22490@hotmail.com }           
          \and
          S. Bagchi \at
          Department of Mathematics\\ 
          National Institute of Technology Durgapur\\
          Burdwan, India. \\
          \email{satya5050@gmail.com} 
}

\date{Received: date / Accepted: date}

\maketitle

\begin{abstract}
Widespread deployment of RFID system arises security and privacy concerns of users. There are several proposals are in the literature to avoid these concerns, but most of them provides reasonable privacy at the cost of search complexity on the server side. The search complexity increases linearly with the number of tags in the system. Some schemes use a group based approach to solve the search complexity problem. In this paper, we proposed a group based authentication protocol for RFID system which is based on some characteristics of cyclic groups. The scheme uses only bitwise XOR and $\mod$ operation for the computational work. Also, the scheme does not use any pseudo-number generator on the tag-side. We use two benchmark metric based on anonymity set to measure the privacy level of the system when some tags are compromised by an adversary. We present some simulation results which show that the scheme preserves high level of privacy and discloses very less amount of information when some tags are compromised. Furthermore, it's formal and informal analysis shows that our scheme preserves information privacy as well as un-traceability and also withstand against various well known attacks.

\end{abstract}
\section{Introduction}

RFID technology is becoming most promising technology in industries to improve the efficiency of tracking and managing goods. Because of its convenience and low-cost, we encounter this technology in various applications like supply chain management, logistics, access control, manufacturing, e-health, passport verification etc \cite{DIMITRIOU2016} \cite{Wang2014} \cite{Chien2006} \cite{Maurya2017}. 

RFID system is made up with three entities: tags, readers, and back-end server. Each tag comprises with a microchip for storing and processing data, and an antenna for receiving and transmitting data.  The server stores all the information about the tags and connected with the readers via a secure channel while the readers communicate with the tags over an insecure channel \cite{Tian2012} \cite{Cho2015} \cite{CAO2016} \cite{Gao2014}. 

With the widespread adoption of RFID system in our daily life, security and privacy concerns are also arise critically \cite{Avoine2014} \cite{SASI2007} \cite{Srivastava2014}. We can use cryptography tool to avoid these concerns but the main obstacle to deploy these tools in RFID system is tight constraints on power, memory and computational capability on the tags.

To enhance the security and privacy of the RFID system and reduce the computational complexity, researchers have been proposed a large number of authentication schemes. Here, we describe some tree-based and group-based authentication schemes related to RFID system. In 2004, Molnar and Wagner \cite{Molnar2004} proposed a tree-based approach for symmetric key authentication scheme. The scheme reduces its identification complexity from linear to logarithmic. However, it violates  the privacy of the other tags when some tags are compromised by an adversary. In 2005, Nohara et al. \cite{Nohara2005} proposed a similar kind of authentication scheme. The scheme provides higher privacy than \cite{Molnar2004} in case of one tag is compromised. Butty{\'a}n et al. \cite{Buttyan2006} proposed an optimal key trees for tree-based private authentication scheme in 2006. They used different branching factors at different levels of the tree to enhancing the privacy level of the scheme. Also, they introduce a benchmark metric for measuring privacy level of the system when some tags are compromised. In 2007, Avoine et al. \cite{Avoine2007} developed a group based symmetric key authentication scheme. They improve the trade-off between scalability and privacy by dividing the tags into a number of groups. Also, they analyze the privacy level of the system by using privacy metric when a single tag is compromised as well as any number of tags are compromised. The main draw back of the scheme is to decrease the privacy level when more tags are compromised. In 2017, Rahman et al. \cite{RAHMAN2017} developed a secure anonymous private authentication scheme which is similar to \cite{Avoine2007} except that the scheme used different technique to provide better privacy and ensure more security. They used privacy metric same as in \cite{Avoine2007} to measure the privacy level of the system. Also, the scheme used information leakage metric based on shannon information theory \cite{Shannon2001} to measure the information leakage in bits of the system when some tags are compromised.

After reviewing the work done, we would like to propose a similar kind of group based authentication scheme. The scheme uses some different kind of techniques to improve privacy and minimize computational cost. 

Rest of the paper is organized as follows: In Section 2, we discuss preliminaries and details of our system model. We present adversary model in Section 3. Group based authentication scheme for RFID system is proposed in Section 4. The formal and informal analysis are given in Section 5. Section 6 illustrates the performance of the proposed scheme. In Section 7, we measure the level of privacy of the system when some tags are compromised by an adversary. We discuss simulation results in Section 8. Finally, conclusions are made in Section 9.

\section{Preliminaries and System Model}
In this section, we give brief overview of a cyclic group \cite{Gallian} and using it's properties, we develop our RFID system model. In this system model, we assume that a reader and the server communicate with each other via a secure communication channel. For simplicity, we assume that the reader and the server are combined into one entity, called reader. 

Suppose $G$ be a nonempty set together with an operation $\ast$ that combines any two elements of $G$ is in $G$. The set $G$ together with this operation is a group if it holds group's law. The order of a group $G$ is the total number of elements in the group. It is denoted by $|G|$. Let $H$ be a subset of a group $G$. We say that $H$ is a subgroup of $G$ if it is itself a group under the operation of $G$.

\begin{defn}
A group $G$ is called cyclic if there exists an element $a\in G$ such that $G=<a>=\{ a^n~|~ n\in \mathbb{Z}\}$. The element $a$ is called a generator of $G$.
\end{defn}
Some important characteristics of cyclic groups are as follows:
\begin{enumerate}
\item Suppose $G=<a>$ be a cyclic group of order $n$. Then $G=<a^k>$ iff $gcd(k,n)=1$.
\item Every subgroup of a cyclic group  is cyclic.
\item Suppose $G=<a>$ be a cyclic group of order $n$. The order of any subgroup of $G$ is a divisor of $n$.
\item For each positive divisor $k$ of $n$, the group $G$ has exactly one subgroup of order $k$ denoted by $<a^{n/k}>$.
\end{enumerate}

\begin{center}
\begin{table}[h]
\begin{threeparttable}
\begin{tabular}{|c|c|c|c|}
\hline

                            Subgroup      & Index          & Tag        & Storage data     \\ \hline

\multirow{4}{*}{$H_1=<a_1>$}           &$i_1$           & $T_{1i_1}$      & $[ID_{1i_1}, K_{1i_1}, r_{1i_{1_{old}}}, r_{1i_{1_{new}}}]$   \\ \cline{2-4}
                                       &$i_2$           & $T_{1i_2}$      & $[ID_{1i_2}, K_{1i_2}, r_{1i_{2_{old}}}, r_{1i_{2_{new}}}]$    \\ \cline{2-4} 
                                       &\vdots          & \vdots        & \vdots    \\ \cline{2-4}
                                       &$i_{|H_1|-1}$   & $T_{1i_{|H_1|-1}}$   & $[ID_{1i_{|H_1|-1}}, K_{1i_{|H_1|-1}}, r_{1i_{(|H_1|-1)_{old}}}, r_{1i_{(|H_1|-1)_{new}}}]$    \\ \hline 

               \vdots                  &\vdots          & \vdots        & \vdots    \\ \hline 

\multirow{4}{*}{$H_m=<a_m>$}           &$i_1$           & $T_{mi_1}$      & $[ID_{mi_1}, K_{mi_1}, r_{mi_{1_{old}}}, r_{mi_{1_{new}}}]$   \\ \cline{2-4}
                                       &$i_2$           & $T_{mi_2}$      & $[ID_{mi_2}, K_{mi_2}, r_{mi_{2_{old}}}, r_{mi_{2_{new}}}]$    \\ \cline{2-4} 
                                       &\vdots          & \vdots        & \vdots    \\ \cline{2-4}
                                       &$i_{|H_m|-1}$   & $T_{mi_{|H_m|-1}}$   & $[ID_{mi_{|H_m|-1}}, K_{mi_{|H_m|-1}}, r_{mi_{(|H_m|-1)_{old}}}, r_{mi_{(|H_m|-1)_{new}}}]$    \\

    \hline  
    
 \vdots                  &\vdots          & \vdots        & \vdots    \\ \hline

\end{tabular}

\end{threeparttable}
\caption{The server look-up table}\label{p6_table1}
\end{table}
\end{center}  

We construct a system model for RFID system with the help of above mentioned characteristics of cyclic groups. For this, we choose a cyclic group $G=<a>$ of order $n$ and find its some subgroups $H_i=<a_i>$, $i=1,2,\dots$, where $a_i= a^{n/k}$ for some positive divisor $k$, according to our requirement using characteristics 4. According to characteristics 4, if $a_i$ and $a_j$ are generators of two different subgroups $H_i$ and $H_j$ respectively, then $a_i$ and $a_j$ are distinct elements in $G$. For the system model, with each element of $H_i$ except the identity element, we assign a tag together with some secret parameters. Additionally, if $H_p$ is the highest order subgroup (say, $|H_p|=P$) and $H_{p-1}$ is the second highest order subgroup (say, $|H_{p-1}|=Q$) in the system. Then we utilize only $Q$ number of elements of $H_p$ to assign tags in such a way so that the value of $i$ (which is used as an index in server look-up table) is same in both the subgroups.    

We made a look-up table for the server which is shown in Table \ref{p6_table1}. From the Table \ref{p6_table1}, we can see that with each element $a_j^i$ in $H_j=<a_j>$, a tag $T_{ji}$ is assigned and $i$ is used as a index in the look-up table for the tag $T_{ji}$. The server stores a unique identification number $ID_{ji}$ and a secret key $K_{ji}$ to the tag $T_{ji}$. It also shows that two nonce $r_{ji_{old}}$ and $r_{ji_{new}}$ are associated with the tag $T_{ji}$. Initially, we take $r_{ji_{old}}=0$ and $r_{ji_{new}}$ be a nonce.

For each tag $T_{ji}$ which is associated with an element $a_j^i$ in $H_j$, we store inverse of $a_j^i$ in $H_j$, i.e. $(a_j^i)^{-1}$, and $i$ in the tag's internal memory. We also store a unique identification number $ID_{ji}$, a secret key $K_{ji}$, and a nonce $R_4$ inside the tag's memory. Initially, $R_4$ is same as $r_{ji_{new}}$ which is stored in the server's look-up table for tag $T_{ji}$.

\section{Adversary Model}
In this section, we present the ability of an adversary $\mathscr{A}$. The adversary is capable to interact with the RFID system $S$ and also, eavesdrops, intercepts, and modifies any transmitted message between any reader and any tag in the system. Our adversarial model is similar to the model proposed by Juels and Weis \cite{Juels2009} with some modifications to meet our requirement. $\mathscr{A}$ is also able to send the following queries to an oracle.
\begin{enumerate}
\item SendTag$(m, T_{ji})\rightarrow m'$ \\
The adversary $\mathscr{A}$ may send a message $m$ to the tag $T_{ji}$ which responds with message $m'$.
\item SendReader$(m, R)\rightarrow m'$ \\ 
 $\mathscr{A}$ can intract with a reader $R$ by sending a message $m$. The reader $R$ responds with message $m'$.
\item DrawTags$(S)$ \\
The adversary has access to a set of tags at any time from the system with this oracle query. 
\item Corrupt$(T_{ji})$ \\
$\mathscr{A}$ is able to access the volatile memory as well as non volatile memory of a tag $T_{ji}$.
\end{enumerate}

We also bound the adversary $\mathscr{A}$ to use SendTag and SendReader queries by $r$ and $t$ respectively. $\mathscr{A}$ can performs $s$ number of computational steps. At a time, $\mathscr{A}$ is able to send Corrupt message to atmost $(n-2)$ number of tags where $n$ is the total number of tags obtained from DrawTags query.

\subsection{Privacy Experiment}
We denote privacy experiment for an RFID system $S$ by $EXP_{\mathscr{A}, S}^{priv}[k,n,r,s,t]$, where $r$, $s$, and $t$ represent the capability of an adversary to use SendTag, computations steps and SendReader respectively. Also, $k$ represents a security parameter. An RFID authentication protocol is considered to be private if no adversary has significant advantage in this experiment.

The main goal of the adversary in the experiment is to distinguish between two different tags with in its computational and interaction limits. The experiment is composed in three phases as follows:
\begin{enumerate}
\item Learning Phase: The adversary $\mathscr{A}$ interacts with the system $S$ and inqueries oracle queries without exceeding its bound and analyze them.
\item Challenging Phase: $\mathscr{A}$ selects two uncorrupted tags from the pool obtained by Drawtags oracle. $\mathscr{A}$ randomly selects any one from them. The adversary evaluates oracles on that particular tag.
\item Guessing Phase: $\mathscr{A}$ outputs a guess bit $b$. $\mathscr{A}$ is expected to produce 1 if he succeeds, otherwise 0.
\end{enumerate}
$EXP$ succeed if $b=1$.

\subsection{Privacy Definition [$(r,s,t)$-privacy]}
According to Juels and Weis \cite{Juels2009}, an RFID authentication protocol with security parameter $k$ is $(r,s,t)$-private if 
$$Pr[EXP_{\mathscr{A}, S}^{priv}[k,n,r,s,t] ~succeeds~ in ~guessing ~b]\leq \frac{1}{2}+\frac{1}{poly(k)},$$ 
where $poly(k)$ is any polynomial function of $k$.




\section{Process}
In this section, we propose a group based authentication protocol which works under all circumstances required for RFID systems. Used notations in this protocol are given in Table \ref{p6_table2}, and proposed protocol is shown in Figure \ref{p6_figure1}. The work flow of the proposed scheme is as follows:





\begin{table}[ht]
\begin{center}
\begin{tabular}{|c l|}  \hline
Notation     & Description  \\ \hline
$G=<a>$      & A finite cyclic group. \\
$H_j=<a_j>$  & A subgroup of the group $G$ with generator element $a_j$. \\
$(a_j^k)^{-1}$ & Inverse element of $a_j^k$ in the subgroup $H_j$. \\
$R_m:m=1, 2, 3$& Nonce generated by the reader.\\
$K_{ji}$     & Secret key of a tag $T_{ji}$ associated with the $i^{th}$ element of a subgroup $H_j$.\\
$ID_{ji}$    & Unique identification number of the tag $T_{ji}$.\\
$\oplus $    & Exclusive-or operation. \\ \hline
\end{tabular}
\caption{Notations and symbols used in proposed scheme}\label{p6_table2}
\end{center}
\end{table}



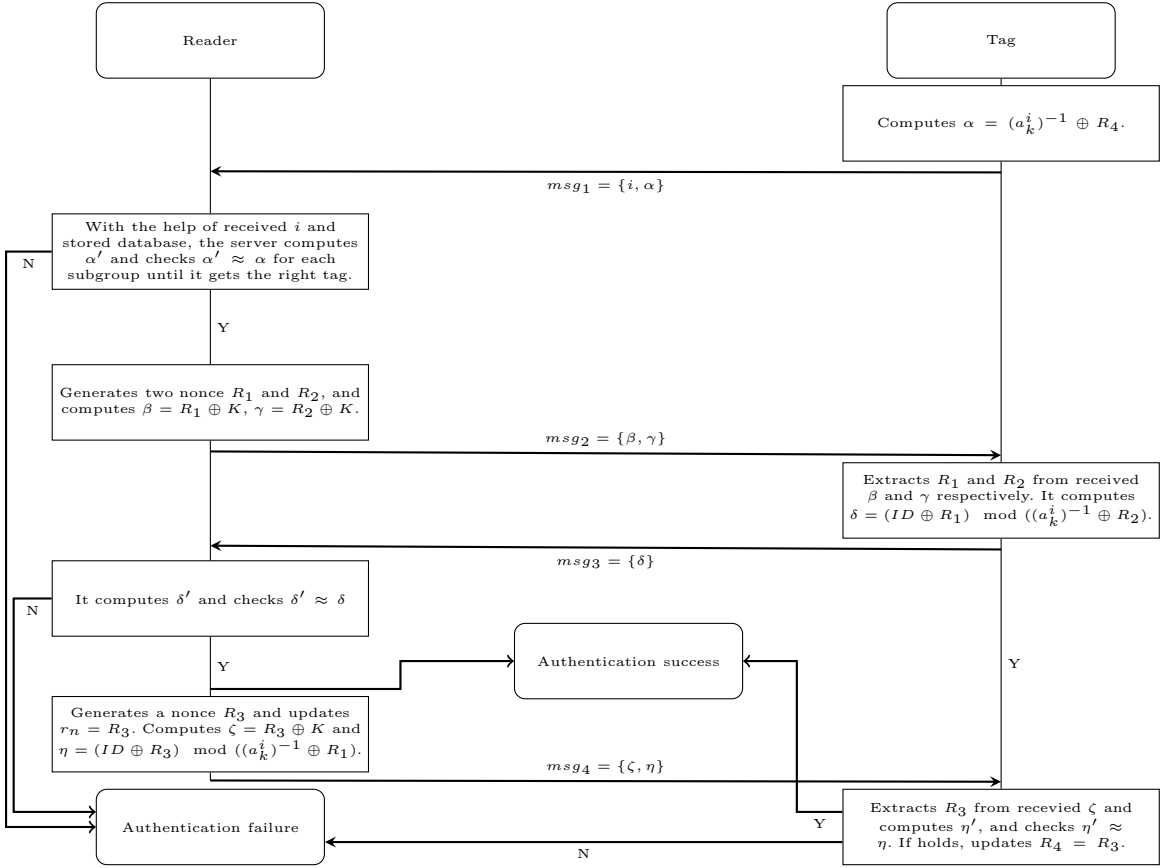
\begin{figure}[t]
\tikzstyle{startstop} = [rectangle, rounded corners, minimum width=3cm, minimum height=1cm,text centered, draw=black,]

\tikzstyle{io} = [trapezium, trapezium left angle=70, trapezium right angle=110, minimum width=3cm, minimum height=1cm, text centered, draw=black,]

\tikzstyle{process} = [rectangle, minimum width=3cm, minimum height=1cm, text centered, text width=4cm, draw=black,]
\tikzstyle{decision} = [diamond, minimum width=.2cm, minimum height=.2cm, text centered,text width=1.4cm, draw=black,]
\tikzstyle{nil}      = []

\tikzstyle{pro} = [rectangle, minimum width=1cm, minimum height=1cm, text centered, text width=2cm, draw=black,]

\tikzstyle{line}      = [draw, -stealth, thick]
\tikzstyle{arrow} = [thick,->,>=stealth]

\begin{tikzpicture}[node distance=5.2cm, auto]
\tiny
\node (reader) [] {};

\node (server) [startstop, left of=reader]  {Reader};
\node [startstop, right of=reader] (tag) {Tag};

%
%
%
%
%
%


%

\node [process, below of=tag, node distance=1.1cm] (process1) {Computes $\alpha=(a^i_k)^{-1}\oplus R_4$.};

\node [process, below of=process1, node distance=5.0cm] (process2) {Extracts $R_1$ and $R_2$ from received $\beta$ and $\gamma$ respectively. It computes $\delta=(ID\oplus R_1)\mod ((a^i_k)^{-1}\oplus R_2)$.};


\node [process, below of=process2, node distance=4.33cm] (process3) {Extracts $R_3$ from recevied $\zeta$ and computes $\eta'$, and checks $\eta'\approx \eta$. If holds, updates $R_4=R_3$.};

\node [startstop, below of=process3, node distance=-2.2cm,xshift=-4.9cm] (process4) {Authentication success};

\node [startstop,  below right of=process3,node distance=0.0cm,xshift=-10.4cm] (process5){Authentication failure};

\draw [-] (tag) -- (process1);
\draw [-] (process1) -- (process2);
\draw [-] (process2) --  node{Y} (process3);


\node [process, below of=server, node distance=2.8cm] (process6) {With the help of received $i$ and stored database, the server computes $\alpha'$ and checks $\alpha'\approx \alpha$ for each subgroup until it gets the right tag.};



\node [process, below of=process6, node distance=2cm] (process7) {Generates two nonce $R_1$ and $R_2$, and computes $\beta=R_1\oplus K$, $\gamma=R_2\oplus K$.};


\node [process, below of=process7, node distance=2.6cm] (process8) {It computes $\delta'$ and checks $\delta'\approx \delta$};
\node [process, below of=process8, node distance=1.8cm] (process9) {Generates a nonce $R_3$ and updates $r_n=R_3$. Computes $\zeta=R_3\oplus K$ and $\eta=(ID\oplus R_3)\mod ((a^i_k)^{-1}\oplus R_1)$.};

\node [nil, below of=process9, node distance=.7cm] (process10) {};

\draw [-] (server) -- (process6);
\draw [-] (process6) -- node{Y}(process7);
\draw [-] (process7) -- (process8);
\draw [-] (process9) -- (process10);

\draw [-] (process8) --node{Y} (process9);

\draw [-stealth, thick,  ->]
   (process6.west)  --
   node {N} ++(-0.6cm,0.0) |-
    (process5.west)
    ;

\draw [-stealth, thick,  ->]
   (process8.west)  --
   node {N} ++(-0.5cm,0.0) |-
    ($(process5.west)+(0,.2)$)
    ;

\draw [-stealth, thick,  ->]
   ($(process8)-(0,1.2)$)  --
    ++(2.5cm,0.0) |-
    (process4.west)
    ;

\draw [-stealth, thick,  ->]
   ($(process3.west)+(0,.2)$)   --node {Y}
   ++(-0.6cm,0.0) |-
    (process4.east)
    ;





%
\draw [arrow]   ($(process1)-(0,.65)$)--node {${msg}_1=\{i,\alpha\}$} ($(process6)+(0,1.065)$);

\draw [arrow]   ($(process7)-(0,.65)$)--node {${msg}_2=\{\beta, \gamma\}$} ($(process2)+(0,.6)$);
\draw [arrow]   ($(process2)-(0,.65)$)--node {${msg}_3=\{\delta\}$} ($(process8)+(0,.7)$);

\draw [arrow]   (process10.north)-- node {${msg}_4=\{\zeta, \eta\}$}($(process3)+(0,.6)$);
\draw [arrow] ($(process3.west)-(0,.2)$) -- node{N}($(process5.east)-(0,.2)$);

%
%
%

\end{tikzpicture}
\caption{Proposed Mutual Authentication Protocol }\label{p6_figure1}
\end{figure}

\begin{enumerate}
\item ${msg}_1 : T_{ki}\rightarrow R : \{i,\alpha\}$ \\
The tag $T_{ki}$ computes $\alpha=(a^i_k)^{-1}\oplus R_4$ and forms a request message ${msg}_1=\{i,\alpha\}$. The tag sends ${msg}_1$ to a reader $R$.

\item ${msg}_2 : R\rightarrow T_{ki}  : \{\beta,\gamma\}$ \\
 After receiving the tag's request message ${msg}_1$, the reader uses $i$ as an index (as in Table \ref{p6_table1}) to performs the following steps for all the subgroups until it finds the right tag:
\begin{enumerate}
\item It calculates the inverse of $a^i_k$ in $H_k$, where $a_k$ is the generator of subgroup $H_k$.
\item The reader computes $\alpha'=(a^i_k)^{-1}\oplus r_{ki_{old}}/r_{ki_{new}}$ and checks whether $\alpha'$ is equal to the received $\alpha$ or not. If so, it gets the right tag $T_{ki}$ (say) inside the subgroup $H_k$. If fails, the reader terminates the protocol.
\item The reader generates two nonce $R_1$ and $R_2$, and computes $\beta=R_1\oplus K_{ki}$, $\gamma=R_2\oplus K_{ki}$, where $K_{ki}$ is the secret key of the tag $T_{ki}$. The reader forms a response message ${msg}_2=\{\beta, \gamma\}$ and transmits it to the tag.
\end{enumerate}

\item ${msg}_3 : T_{ki}\rightarrow R  : \{\delta\}$ \\
Upon receiving the response message ${msg}_2$, the tag $T_{ki}$ extracts $R_1$ and $R_2$ from $\beta$ and $\gamma$ respectively with the help of its secret key $K_{ki}$. It computes $\delta= (ID_{ki}\oplus R_1) \mod ((a^i_k)^{-1}\oplus R_2)$ and send $\delta$ inside the response message $msg_3=\{\delta$\} to the reader.

\item ${msg}_4 : R\rightarrow T_{ki}  : \{\zeta, \eta\}$ \\
After receiving message ${msg}_3$ from the tag $T_{ki}$, the reader calculates $\delta'= (ID_{ki}\oplus R_1) \mod ((a^i_k)^{-1}\oplus R_2)$ for the tag $T_{ki}$ and checks whether $\delta'$ is equal to the received $\delta$ or not. If it holds, the reader authenticates the tag $T_{ki}$ otherwise terminates the session. If tag's authentication succeed, the reader generates a nonce $R_3$ and assigns $r_{ki_{old}}=r_{ki_{new}}$, and $r_{ki_{new}}=R_3$. Simultaneously, also computes $\zeta=R_3\oplus K_{ki}$, $\eta=(ID_{ki}\oplus R_3) \mod ((a^i_k)^{-1}\oplus R_1)$ and forms a response message $msg_4=\{\zeta, \eta\}$. It sends $msg_4$ to the tag $T_{ki}$. 

\item  
Upon receiving message $msg_4$, the tag $T_{ki}$ extracts $R_3$ from $\zeta$ and calculates $\eta'=(ID_{ki}\oplus R_3) \mod ((a^i_k)^{-1}\oplus R_1)$. The tag checks whether $\eta'$ is equal to the received $\eta$ or not. If so, the tag authenticates the reader and updates $R_4=R_3$ for further communication. 
\end{enumerate}

\section{Security and Privacy Analysis}
In this section, we present formal and informal analysis of our proposed scheme with respect to above mentioned adversary model. The formal analysis shows that the proposed scheme preserves privacy and un-traceability. Also, it's informal analysis shows that the proposed scheme is secure against various well-known attacks.
\subsection{Formal Security Analysis}
\begin{theorem}
The proposed scheme attains information privacy with respect to a adversary $\mathscr{A}$. 
\end{theorem}
\begin{proof}
We assume that the proposed scheme does not preserve information privacy. So the success probability of the adversary to win experiment is non-negligible. $\mathscr{A}$'s privacy game is composed in three phases as follows:
\begin{itemize}
\item Learning Phase: The adversary gets a set of $n$-tags by querying DrawTags oracle. $\mathscr{A}$ can send any oracle queries to a tag $T_i$ (say) without exceeding its computation bound and analyze them. $\mathscr{A}$ can use Corrupt oracle to atmost $n-2$ tags.
\begin{center}
     $T_i \leftarrow$ DrawTag$(S)$

     $msg_1=\{i, \alpha\}\leftarrow$ SendTag$(init, T_i)$

     $msg_2=\{\beta, \gamma\}\leftarrow$ SendReader$(msg_1, R)$

     $msg_3=\{i, \delta\}\leftarrow$ SendTag$(msg_2, T_i)$

     $msg_4=\{\zeta, \eta\}\leftarrow$ SendReader$(msg_3, R)$

    $\{i, (a^i_j)^{-1}, K, ID, R\}\leftarrow$ Corrupt$(T_i).$
\end{center}
\item Challenge Phase: The adversary $\mathscr{A}$ selects two uncorrupted tags say, $T_i$ and $T_j$, from the set of tags obtained by DrawTags query as its challenge tags. Let $T_0^*=T_i$, $T_1^*=T_j$, and $b \in \{0, 1\}$. $\mathscr{A}$ randomly selects $T_b$ among them and analyze all queries run on it. Note that $\mathscr{A}$ is not able to use Corrupt oracle on that particular tag $T_b$. 
\begin{center}
     
     $msg_1=\{i, \alpha\}\leftarrow$ SendTag$(init, T_b)$

     $msg_2=\{\beta, \gamma\}\leftarrow$ SendReader$(msg_1, R)$

     $msg_3=\{i, \delta\}\leftarrow$ SendTag$(msg_2, T_b)$

     $msg_4=\{\zeta, \eta\}\leftarrow$ SendReader$(msg_3, R).$

    \end{center}
\item Guess Phase: Eventually, the adversary  outputs a guess bit $b'$ for the corresponding tag. 
\end{itemize} 
 The adversary wins the experiment if $b'=b$. It is possible only when the adversary knows all the secrets stored in $T_b$'s internal memory as well as the mother group $G$. So our assumption is wrong. Hence the proposed scheme preserves the information privacy with respect to $\mathscr{A}$. 
\end{proof}

\begin{theorem}
The proposed scheme provides un-traceability with respect to the adversary $\mathscr{A}$. 
\end{theorem}
\begin{proof}
Let us assume that the proposed scheme is traceable. i.e. the adversary can trace a tag at any time. This means $\mathscr{A}$ is able to distinguish  between two tags. We show that our assumption is wrong with the help of $\mathscr{A}$'s privacy game which is as follows:
\begin{itemize}
\item Learning Phase: $\mathscr{A}$ uses DrawTags query for the system $S$ and gets access to $n$-tags. For all the tags, $\mathscr{A}$ sends SendTag and SendReader queries to get transmitted information among a reader and tags. The adversary analyzes all the transmitted message. The adversary can use Corrupt query for atmost $n-2$ tags because the goal of the privacy game is to distinguish between two uncorrupted tags.
\begin{center}
     $T_i \leftarrow$ DrawTag$(S)$

     $msg_1=\{i, \alpha\}\leftarrow$ SendTag$(init, T_i)$

     $msg_2=\{\beta, \gamma\}\leftarrow$ SendReader$(msg_1, R)$

     $msg_3=\{i, \delta\}\leftarrow$ SendTag$(msg_2, T_i)$

     $msg_4=\{\zeta, \eta\}\leftarrow$ SendReader$(msg_3, R)$

    $\{i, (a^i_j)^{-1}, K, ID, R\}\leftarrow$ Corrupt$(T_i).$
\end{center}
\item Challenge Phase: The adversary selects two uncorrupted tags $T_i$ and $T_j$ to which it did not send Corrupt query in the learning phase. $\mathscr{A}$ randomly selects $T_b: b \in \{i, j\}$ among them. The adversary queries all the oracle queries except Corrupt query to the tag $T_b$ and evaluates them.
\begin{center}
     
     $msg_1^*=\{i, \alpha\}\leftarrow$ SendTag$(init, T_b)$

     $msg_2^*=\{\beta, \gamma\}\leftarrow$ SendReader$(msg_1^*, R)$

     $msg_3^*=\{i, \delta\}\leftarrow$ SendTag$(msg_2^*, T_b)$

     $msg_4^*=\{\zeta, \eta\}\leftarrow$ SendReader$(msg_3^*, R).$
\end{center} 

\item Guess Phase: $\mathscr{A}$ outputs a guess bit $b'$.   
\end{itemize}
The adversary wins the game if $b'=b$ but it is possible only when if $$Pr[msg_1^*=msg_1]=1$$ Since the message $msg_1$ depend upon the tag's nonce $R_4$ which is different in each protocol run. So our assumption is wrong. Hence the adversary is unable to trace the tag.
\end{proof}

\subsection{Informal Security Analysis}
\subsubsection{Replay Attack Resistance}
An adversary can eavesdrops the wireless channel and keeps all the transmitted messages between a reader and a tag. The adversary uses these message into another session to disguise itself as the tag or the reader to deceive the other one. In the proposed scheme, it is infeasible for an adversary to forge messages as a valid tag/reader because each transmitted message incorporates a fresh nonce in each authentication session which can not be get by the adversary (since the nonce XOR with some other secret parameters). This makes all the replayed message by the adversary are illegal message. Thus the scheme prevents strongly the replay attack.
\subsubsection{De-synchronization Attack Resistance}
For each tag, the server stores two nonce $r_{old}$ and $r_{new}$ in its database to save the scheme from the de-synchronization attack. The server also updates these values after a successful authentication session. An adversary intercepts or modified any transmitted message in one session in such a way so that a tag does not update the value of stored nonce. The server can authenticate the legitimate tag by its old value stored in database into another session. So it is not possible for an adversary to de-synchronize the scheme.
\subsubsection{Man-in-Middle Attack Resistance}
An adversary is unable to act as a middle man in between a reader and a tag because it is infeasible for the adversary to intercepts any transmitted message without knowing the secret key, unique identification number, and knowledge about the cyclic group. The probability of guessing or calculating these values from the transmitted message is negligible because a fresh nonce is used in each transmitted message.   

\section{Performance Analysis}
In this section, we present efficiency of our proposed scheme in terms of tag computation, server computation, and storage, as described in Table \ref{p6_table3}. The proposed scheme's search complexity is $O(\gamma)$ which is same as in Avione et al. \cite{Avoine2007} but better than Rahman et al. \cite{RAHMAN2017}. During the authentication phase, the scheme performs only bit-wise XOR operation whereas schemes of Avoine et al. and Rahman et al. perform symmetric key encryption and decryption. Also, the proposed scheme does not use any pseudo number generator function for generating nonce on the tag-side. It uses nonce generated by the reader. We assume that all the parameters used in the proposed scheme are $L$-bits long. On the tag side, our scheme keeps five items. Thus the storage cost is $5L$ bits. The proposed scheme also provides mutual authentication among a reader and tags. When we compare with \cite{Avoine2007} and \cite{RAHMAN2017} in terms of computation, the proposed scheme performs very less computation which is optimal for the real world tiny powered tags.
\begin{center}
\begin{table}[h]
\begin{tabular}{|l|c|c|c|c|}
\hline
Protocol                              & Entity            & Avoine               & Rahman                & Proposed protocol \\ 
                                      &                   & \cite{Avoine2007}     & \cite{RAHMAN2017}       &             \\ \hline

Symmetric encryption/                &T                   & 2                    & 2                           &   $\times$ \\ \cline{2-5}
decryption                           &R                   & 2                    & 2                           &   $\times$ \\ \hline

Search complexity                    &R                   & $O(\gamma)$          & $O(\gamma +|\pi|)$           & $O(\gamma)$\\ \hline

No. of PRNG                          &T                   & 1                    & 1                           & $\times$ \\ \hline


Required memory                      & T                 & $3L$                  &$(m+2)L$                 & $5L$\\ \hline

Mutual authentication                &                   &  $\times$             &  $\times$               & \checkmark \\ \hline

\end{tabular}
\begin{tablenotes}
    \item[] $\gamma$- Total number of groups in the system.
    \item[] $|\pi|$ - Total number of secret keys of a tag associated with the identifier $ID_x$.
    \item[] $m$ - Number of identifier is assigned to each tag.
  \end{tablenotes}
\caption{Computation cost performance comparison}\label{p6_table3}
\end{table}
\end{center}


\section{Measurement of Privacy}
In this section, we analyze the privacy level of our proposed scheme in terms of anonymity set and data leakage in bits. For the anonymity sets, we use privacy metric introdued by \cite{Buttyan2006}. Also, we use another metric says information leakage for data leakage proposed by shannon \cite{Shannon2001} and used in \cite{Nohl2006} \cite{RAHMAN2017} to measure the information (in bits) disclosed by the proposed scheme when some tags are compromised.

Both the metric use disjoint partition sets of tags for observation. When some tags are compromised, the set of all tags are partitioned in such a way so that the adversary can not distinguish the tags that belong to the same partition but she can distinguish the tags belong to different partitions. Here, $|P_i|$ denotes the size of such partition $P_i$ and $\frac{|P_i|}{N}$ is the probability that a randomly chosen tag belongs to partition $P_i$.

\subsection{Level of privacy based on anonymity set}

The level of privacy $\Re$ based on anonymity set is characterized as average anonymity set size normalized with the total number of tags $N$ \cite{Buttyan2006} \cite{Avoine2007} \cite{RAHMAN2017}.
\begin{equation}
       \Re=\frac{1}{N}\sum_{i} |P_i|\frac{|P_i|}{N}= \frac{1}{N^2}\sum_{i} |P_i|^2. \label{eq1}
\end{equation}
In the proposed scheme, if a tag is compromised, it does not leak any information about the subgroup in which it belongs. For this reason, the adversary can not distinguish between two tags whether they belongs to same subgroup or not. So, if $\mathbb{C}$ is the total number of compromised tags in the whole system, we partitioned the system into $\mathbb{C}$ number of anonymity sets with size $1$ and one another anonymity set of size $(N-\mathbb{C})$.
Using equation \ref{eq1}, the level of privacy $\Re$ achieved by our scheme is
\begin{equation}
\Re = \frac{1}{N^2}\{\mathbb{C}+(N-\mathbb{C})^2\}, \label{eq2}
\end{equation}
where $N$ is the total number of tags in the system and $\mathbb{C}$ is the total number of compromised tags in the system.

\subsection{Level of privacy based on information leakage in bits}

According to Rahman et al. \cite{RAHMAN2017}, if an adversary partitioned a system with $N$ tags into $k$ disjoint sets, then the information leakage in bits can be expressed as follows:
\begin{equation}
 \mathbb{I}=\sum_{i=1}^{k}\frac{|P_i|}{N}\log_2 \left(\frac{N}{|P_i|}\right). \label{eq3}
\end{equation}
In the proposed scheme, if $\mathbb{C}$ is the total number of compromised tags in the system. Then we partitioned the system with $N$ tags into $\mathbb{C}$ anonymity sets of size $1$ and one another anonymity set of size $(N-\mathbb{C})$. According to our partitions, the information leakage in bits  is as follows
\begin{equation}
 \mathbb{I}=\frac{\mathbb{C}}{N}\log_2 N + \frac{(N-\mathbb{C})}{N}\log_2 \left(\frac{N}{N-\mathbb{C}}\right). \label{eq4}  
\end{equation}
\section{Experimental Results}
\begin{figure}[t]
\begin{center}
\includegraphics[height=2.5in,width=5.0in]{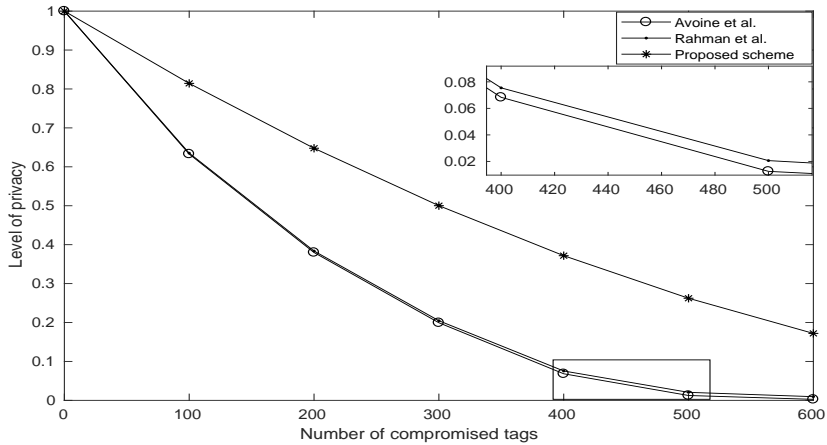}\\
\caption{Level of privacy of the system based on anonymity set}\label{p6_figure2}
\end{center}
\end{figure}

\begin{figure}[t]
\begin{center}
\includegraphics[height=2.5in,width=5.0in]{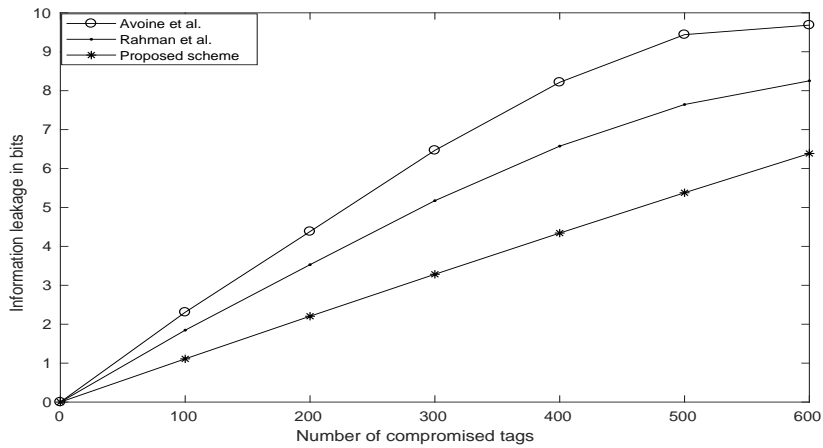}\\
\caption{Level of privacy of the system based on information leakage in bits}\label{p6_figure3}
\end{center}
\end{figure}

In this section, we compare our scheme with Avoine et al. \cite{Avoine2007} and Rahman et al. \cite{RAHMAN2017} using a matlab simulation. The simulation is done using the expressions (\ref{eq1}) - (\ref{eq4}). In the simulation, we assume that the system has $N=2^{10}$ number of tags and all the tags are divided into 32 groups. We choose range of compromised tags from 0 to 600. In the proposed scheme, it is not necessary to take same number of tags in each groups. In the simulation, we run 100 simulations for each value of compromised tags $\mathbb{C}$ in the system. In each simulation run, compromised tags are chosen uniformly random from the groups of all tags. Finally, we average all the obtained values over all simulation runs.  The simulation results are shown in Figure \ref{p6_figure2} and Figure \ref{p6_figure3}. The simulation results of the Figure \ref{p6_figure2} shows that the privacy level achieved by the proposed scheme is $94.42\%$ and $98.43\%$ better than Rahman et al. and Avoine et al. respectively, when $\mathbb{C}$ becomes 600 in a similar setup. According to simulation result shown in Figure \ref{p6_figure3}, the proposed scheme discloses $22.62\%$ and $34.05\%$ less information than Rahman et al. and Avoine et al. respectively when $\mathbb{C}$ becomes 600. Thus the proposed scheme achieves higher improvement in terms of privacy level and information leakage than the other schemes, when some tags are compromised by an adversary.


\section{Conclusion}
In this paper, we have proposed a group based authentication scheme for RFID system based on a cyclic group. The detailed formal analysis shows that it preserves information privacy and un-traceability. The informal analysis shows that the scheme resists various existing attacks. The performance analysis illustrates that the scheme uses very less resources on tags to performs computational work and storage data. The experimental results show that our scheme preserves high level privacy when some tags are compromised. Thus, the analysis and prominent features conclude that the scheme is secure and efficient for a low-cost RFID system. 

\textbf{Conflict of Interest:} The author P K Maurya thanks to MHRD, India, for financial support of his research.



\end{document}